\documentclass[a4paper,reqno]{amsart}
\usepackage{a4wide}
\usepackage{amsmath,amsfonts,amssymb,amsthm,color}

\usepackage{enumitem}

\theoremstyle{plain}
\newtheorem{theorem}{Theorem}[section]

\newtheorem{lemma}[theorem]{Lemma}

\theoremstyle{remark}
\newtheorem{remark}[theorem]{Remark}
\newtheorem{definition}[theorem]{Definition}

\newcommand{\Z}{\mathbb{Z}}
\newcommand{\R}{\mathbb{R}}
\newcommand{\C}{\mathbb{C}}
\newcommand{\eu}{\mathrm{e}}
\newcommand{\iu}{\mathrm{i}}
\newcommand{\di}{\mathrm{d}}
\newcommand{\bP}{\mathbf{P}}

\newcommand{\bx}{\mathbf{x}}
\newcommand{\by}{\mathbf{y}}
\newcommand{\Id}{\mathbf{1}}
\newcommand{\cI}{\mathcal{I}}

\newcommand{\set}[1]{\left\{ #1 \right\}}
\newcommand{\su}[1]{^{\mathrm{#1}}}
\newcommand{\sub}[1]{_{\mathrm{#1}}}
\newcommand{\up}{\uparrow}
\newcommand{\down}{\downarrow}
\newcommand{\Tr}{\mathrm{Tr}}
\newcommand{\IsDOS}{\mathrm{IsDOS}}
\newcommand{\SCh}{\mathrm{SCh}}
\newcommand{\Ch}{\mathrm{Ch}}

\title{St\v{r}eda formula for charge and spin currents}
\author{Domenico Monaco \and Massimo Moscolari}
\date{May 11, 2020. arXiv version 2} 

\usepackage{hyperref}

\begin{document}

\begin{abstract}
We consider a 2-dimensional Bloch--Landau--Pauli Hamiltonian for a spinful electron in a constant magnetic field subject to a periodic background potential. Assuming that the $z$-component of the spin operator is conserved, we compute the linear response of the associated spin density of states to a small change in the magnetic field, and identify it with the spin Hall conductivity. This response is in the form of a spin Chern marker, which is in general quantized to a half-integer, and to an integer under the further assumption of time-reversal symmetry. Our result is thus a generalization to the context of the quantum spin Hall effect to the well-known formula by St\v{r}eda, which is formulated instead for charge transport.

\bigskip


\medskip

\end{abstract}

\maketitle

\tableofcontents

\section{Introduction}

The study of spin transport and the development of ``spintronic'' devices has recently attracted a lot of attention in the condensed matter community \cite{Schliemann06, JungwirthWunderlichOlejnik12, Sinova15}, which in turn initiated the mathematical analysis of spin transport \cite{Prodan09, SchulzBaldes, MarcelliTesi, MarcelliPanatiTauber, MarcelliPanatiTeufel}. A major part of the scientific activity in this direction has sparked after the discovery of time-reversal symmetric topological insulators, which can host the analogue of the quantum Hall effect but in the context of spin currents, namely the {\it quantum spin Hall effect} \cite{KM05, BHZ06, HasanKane10, Molenkamp10}.

The interest in the quantum Hall effect lies in the fact that a transport coefficient, namely the transverse charge conductivity $\sigma_{12}$, appears to be quantized (in appropriate physical units) up to an astounding precision: this was explained by relating this conductivity to a topological object, commonly known as the Chern number, which can only take integer values. One possible way to realize this connection was proposed by St\v{r}eda in \cite{Streda82a, Streda82}, by equating $\sigma_{12}$ to the variation with respect to an external magnetic field of the integrated density of states of the system, when the Fermi energy is assumed to stay in a spectral gap of the underlying Hamiltonian. Later, this and related results were proved in several contexts by different groups in the mathematical physics community: we refer to \cite{CorneanMonacoMoscolari18} for a recent account on this research line. For our discussion, it is nonetheless worth stressing that the St\v{r}eda formula was proved by Cornean, Nenciu and Pedersen in \cite{CorneanNenciuPedersen06} in the framework that we will adopt as well, namely that of gapped continuum Schr\"{o}dinger operators in 2-dimensions modelling the dynamics of a spin-$1/2$ particle in a magnetic field subject to a $\Z^2$-periodic background potential.

Our goal is to modify and adapt the proof of the St\v{r}eda formula provided by \cite{CorneanNenciuPedersen06} to accommodate for a {\it conserved} spin degree of freedom. We will thus relate the spin Hall conductivity to the derivative with respect to the external magnetic field of the spin density of states, and obtain an expression for this derivative in terms of what we call a {\it spin Chern marker} (inspired by the topological markers in position space defined in \cite{Prodan09,MarcelliMonacoMoscolariPanati,BiancoResta2011,Caio et al 2019}). Our contribution can thus be seen as a formal proof of ``spin St\v{r}eda formulas'' appearing in the physics literature \cite{YangChang06, Murakami06}, and as a validation for the spin Chern number proposed by Prodan to investigate the quantum spin Hall topological phase \cite{Prodan09}, at least in a setting where the $z$-component of the spin operator is conserved.

The present paper fits in a broader research line focused on establishing the mathematical theory of spin trasport (see \cite{MarcelliTesi, MarcelliPanatiTauber, MarcelliPanatiTeufel} and references therein). For example, some of the techniques from perturbation theory we employ have been used and generalized in the latter references to conduct a careful study of the linear response of a spin current to an external electric field, including the situation -- which is not covered by the present work -- of models which \emph{do not} preserve the spin degree of freedom. To be more specific, it is worth mentioning that the definition of a spin current (and spin Hall conductivity) depends in general on the form of the quantum mechanical spin current operator. A conventional definition for the latter is ${\mathbf {J}}^z= -\frac{1}{2} \big\{ \dot{{\bf X}}, S^z \big\}$, where $-\dot{{\bf X}} = -\iu \left[H, {\bf X}\right]$ is the standard (charge) current operator associated to the Hamiltonian $H$ for a particle with charge $q=-1$, while $S^z$ denotes the $z$-component of the spin operator. On the basis of a ``mesoscopic'' continuity equation, it was proposed by Niu and collaborators in \cite{Niu06} that the ``proper'' definition should instead read $\mathbf{J}^z\sub{proper} = -\iu \big [H, \mathbf{X} S^z \big]$, which leads to well-posed spin currents also in the case of spin non-conserving Hamiltonians, but has the disadvantage of not being a translation-invariant operator. The use of this ``proper'' spin current operator requires setting up a new framework to study (spin) transport, and leads to more involved linear response formulas for the spin Hall conductivity \cite{MarcelliPanatiTeufel}. The distinction between the two current operators has also led to a debate in the present context of the spin St\v{r}eda formula: \cite{YangChang06} adopts the ``conventional'' definition while \cite{Murakami06} argues on the basis of the ``proper'' spin current operator. When $H$ commutes with $S^z$, as we will assume, the two definitions coincide, and no ambiguity arises. We still consider the derivation of the spin St\v{r}eda formula for spin non-preserving Hamiltonians an interesting research line, which we postpone to future work.

Although we formulate the results mainly for spin currents, if one substitutes the spin charge operator $-S^z$ with the electric charge operator $-\Id$, our proofs go through unchanged, and our arguments (heavily inspired to \cite{CorneanNenciuPedersen06}) can be used to cover charge currents and the quantum Hall setting as well. In particular, we obtain a partly new proof (presented with a slightly different strategy in \cite{Moscolari}) of the quantization of the Hall conductivity as a multiple of an integer Chern number. 

\medskip

\noindent \textsc{Acknowledgments.} The authors are grateful to Horia Cornean, S{\o}ren Fournais, and Jacob Schach~M{\o}ller for organizing the QMath14 workshop, and to Gianluca Panati and Marcello Porta for the invitation to participate in the Condensed Matter session. This work benefited from the discussions of the authors with Horia Cornean, Giovanna Marcelli, Gianluca Panati, and Stefan Teufel on topics related to the mathematics of spin transport. The work of D.\ M.\ has been supported by the European Research Council (ERC) under the European Union’s Horizon 2020 research and innovation programme (ERC CoG UniCoSM, grant agreement n.724939). The work of M.\ M.\ is supported by the Grant 8021-00084B of the Danish Council for Independent Research~$|$~Natural Sciences.

\section{Generalities on Bloch--Landau--Pauli Hamiltonians}

Following \cite{CorneanNenciuPedersen06}, our reference Bloch--Landau--Pauli Hamiltonian acts on $L^2(\R^2) \otimes \C^2$ as (in appropriate units)
\begin{equation}
\label{HamiltonianB_1B_2}
H_{B_1,B_2} := \frac{1}{2} \, \bP(B_2)^2 + V(\bx) \otimes \Id + 2 \, B_1 \, \Id \otimes s^z, \quad \bx = (x_1,x_2) \in \R^2,
\end{equation}
where
\[ \bP(B) := \left( - \iu \nabla_{\bx} - \frac{B}{2} \, (-x_2,x_1) \right) \otimes \Id + \frac{1}{2} \left( - \partial_{x_2} V(\bx) , \partial_{x_1} V(\bx) \right) \otimes s^z \]
and $s^z := \sigma^z/2$ is the $z$-component of the spin operator (half of the third Pauli matrix). The magnetic momentum $\bP(B)$ is modified to take into account the intrinsic spin-orbit coupling \cite{LuttingerKohn55}. Physically, the parameters $B_1$ and $B_2$ are modelling the same magnetic field and should be equal, but for the sake of a more general discussion this is not needed for the moment; we will later simplify notation setting $B_1 = B_2 = B$ when appropriate. Observe in particular that
\begin{equation} \label{sz is conserved}
\big[ P_j(B_2), S^z \big] = \big[ H_{B_1,B_2}, S^z \big] = 0, \quad S^z := \Id \otimes s^z, \quad j \in \set{1,2}.
\end{equation}

We assume $V$ to be a smooth, real-valued, $\Z^2$-periodic potential. With this hypothesis, the Hamiltonian $H_{B_1, B_2}$ is a selfadjoint operator bounded from below, and $C^{\infty}_0(\R^2)\otimes\mathbb{C}^2$ is a dense core. Moreover, we assume that $H_{B_1, B_2}$ has a spectral gap around a certain (Fermi) energy $E\sub{F}$. We denote by $\Pi_{B_1, B_2}$ the corresponding Fermi projection, defined by the following Riesz formula in terms of the resolvent of $H_{B_1,B_2}$ \cite{Kato66}:
\begin{equation} \label{Riesz}
\Pi_{B_1, B_2} = \frac{\iu}{2\pi} \, \oint_\Gamma \di w \, \big( H_{B_1,B_2} - w \Id \big)^{-1},
\end{equation}
where $\Gamma$ is a positively-oriented contour in the complex energy plane surrounding all the bands below $E\sub{F}$ and contained in the resolvent set of $H_{B_1,B_2}$. Notice that the terms in the Hamiltonian which are proportional to $S^z$ produce only bounded perturbations, therefore we can apply general results from the theory of magnetic Schr\"odinger operators \cite{Simon,CorneanNenciu09} and conclude that the resolvent operator has an integral kernel that decays away from the diagonal and has a logarithmic singularity on the diagonal \cite{Cornean10, BroderixHundertmarkLeschke}, that is,
$$
\sup_{w \in \Gamma} \left\|\big( H_{B_1,B_2} - w \Id \big)^{-1}({\bf x};{\bf y})\right\| \leq C_{\Gamma} \left(2 + \ln (1+|{\bf x}-{\bf y}|^{-1}) \right)\eu^{-\beta_{\Gamma}|{\bf x}-{\bf y}|},
$$
where $C_{\Gamma}$ and $\beta_{\Gamma}$ are positive constants that depend only on the contour $\Gamma$, and where the norm on the left-hand side is the matrix norm in $\C^2$. This implies that $\Pi_{B_1,B_2}$ has a jointly continuous integral kernel with off-diagonal exponential decay, meaning that
\[ \left\| \Pi_{B_1,B_2}(\bx;\by) \right\| \le C \, \eu^{- \alpha \, |\bx - \by|}, \]
for some positive constants $C, \alpha > 0$. Notice that the exponential decay is a consequence of Combes--Thomas estimates \cite{CT73} combined with the existence of the spectral gap. The same is then true for the operators $[X_j, \Pi_{B_1,B_2}]$, $j \in \set{1,2}$, where $X_j$ denotes the position operator in the $j$-th direction: indeed
\[ [X_j, \Pi_{B_1,B_2}](\bx;\by) = (x_j - y_j) \, \Pi_{B_1,B_2}(\bx;\by) \]
as an identity between $2 \times 2$ matrices. In particular,  the off-diagonal decay of the integral kernel of $[X_j, \Pi_{B_1,B_2}]$ implies via Schur's criterion that it is a bounded operator. Furthermore, all these operators are {\it magnetic covariant}, in the sense that their integral kernels satisfy 
\[ \eu^{\iu B_2\, (x_2 \, n_1 - x_1 n_2)/2} K(\bx-\mathbf{n};\by-\mathbf{n}) \eu^{-\iu B_2\, (y_2 \, n_1 -yx_1 n_2)/2}=K(\bx;\by) \quad \text{ for all } \mathbf{n} \in \Z^2, \]
which is a crucial property in the computation of certain thermodynamic observables (cf.\ Section~\ref{sec:SHC}).

By {\it gauge-covariant magnetic perturbation theory} \cite{CorneanNenciu98,Nenciu02,CorneanNenciu09,Cornean10,CorneanMonacoMoscolari18}, the spectral gap around $E\sub{F}$ persists also for neighbouring values of $B_1$ and $B_2$, and so we can consider $\Pi_{B_1,B_2}$ as a function of $B_1, B_2$. In general, due to the singularity of the perturbation given by terms proportional to $B_2$, this will be very ill-behaved (not even continuous) if we look for example at the norm-topology of bounded operators on $L^2(\R^2) \otimes \C^2$, see for example \cite[Corollary 1.2]{CorneanMonacoMoscolari18}. Nonetheless, there are certain other functions of $\Pi_{B_1,B_2}$, expressed in terms of its integral kernel, which are much more regular. For our purposes, the most relevant one is related to its {\it trace per unit volume}, defined for any magnetic covariant operator $A$ with localized kernel as
\[ \tau(A) := \frac{1}{|\Omega|} \int_{\Omega} \di \bx \, \Tr_{\C^2} \big( A(\bx;\bx) \big) = \frac{1}{|\Omega|} \int_{\Omega} \di \bx \, \sum_{s \in \set{\up,\down}} A(\bx,s; \bx,s) \]
where $\Omega = (0,1)^2$ is the unit cell for the $\Z^2$-periodic potential. As a general remark, we will often use that this trace-like functional is {\it cyclic} on magnetic covariant operators with kernels that are exponentially localized away from the diagonal, namely that if $A, B$ are such operators then $\tau(A \, B) = \tau(B \, A)$ (cf.~\cite[Eqn.~(4.25)]{CorneanNenciuPedersen06}).

\begin{definition}
	The {\it integrated spin density of states} is defined as
	$$ \IsDOS(B_1,B_2) := \tau \left( S^z \, \Pi_{B_1, B_2} \right) = \frac{1}{2} \int_{\Omega} \di \bx \, \big\{ \Pi_{B_1, B_2}(\bx,\up; \bx,\up) - \Pi_{B_1, B_2}(\bx,\down; \bx,\down) \big\}. $$
\end{definition} 

Magnetic perturbation theory allows to show that $\IsDOS(B_1,B_2)$ is a smooth function of $B_1, B_2$: indeed, one can compute for example that (cf.~\cite[Eqn.~(4.23)]{CorneanNenciuPedersen06})
\begin{equation} \label{4.21}
\begin{aligned}
&\left.\frac{\partial}{\partial B_2} \, \IsDOS(B_1,B_2)  \right|_{B_1 = B_2 = B} \\
& = - \frac{1}{4\pi} \tau \left( \int_{\Gamma} \di w \, S^z \, R_w(B) \, P_1(B) \, R_w(B) \, P_2(B) \, R_w(B) - \big( 1 \leftrightarrow 2 \big) \right)
\end{aligned}
\end{equation}
where $\Gamma$ is the usual complex-energy contour appearing in~\eqref{Riesz} and
$ R_w(B) := \big( H_{B_1=B,B_2=B} - w \Id \big)^{-1}$. Note in particular that, to realize that the operators in the trace per unit volume have an exponentially localized integral kernel, one has to exploit the complex integral with respect to $w$ \cite{CorneanNenciu09}.

Notice that only the perturbation due to the magnetic field $B_2$ is singular and has to be treated using gauge covariant magnetic perturbation theory, while the perturbation of the term proportional to $B_1$ is actually bounded and can be handled with the standard techniques of linear perturbation theory \cite{Kato66}. Therefore, for $\lambda:=|\widetilde{B}-B_1|$ small enough, we have that 
$$
\left\|\Pi_{B_1, B_2}-\Pi_{\widetilde{B},B_2}\right\| \leq \lambda C,
$$
for some positive constant $C$. Hence, arguing e.g.\ as in \cite[Lemma C.1]{CorneanMonacoMoscolari18}, one can show that $\IsDOS(B_1,B_2)=\IsDOS(\widetilde{B},B_2)$. We can conclude that the partial derivative of the $\IsDOS(B_1,B_2)$ with respect to $B_2$ evaluated in $B_1=B_2=B$ is actually equal to the total derivative of $\IsDOS(B,B)=:\IsDOS(B)$ with respect to $B$. Not to overburden the notation, we therefore set $B_1 = B_2 = B$ from here on out.

\begin{remark} \label{remark}
Using the fact that $\Pi_{B}:=\Pi_{B,B}$ is a projector, so that in particular $\Pi_{B}^2 = \Pi_{B}$, \eqref{4.21} can be equivalently written as
\begin{align*}
& \frac{\di}{\di B} \, \tau \left( S^z \, \Pi_{B} \right) =  \frac{\di}{\di B} \, \tau \left( S^z \, \Pi_{B}^2 \right) \\
& \qquad = - \frac{1}{4\pi} \, \tau \Bigg( \int_{\Gamma} \di w \, S^z \, \Pi_B \, R_w(B) \, P_1(B) \, R_w(B) \, P_2(B) \, R_w(B) \\
& \qquad\qquad + S^z \, R_w(B) \, P_1(B) \, R_w(B) \, P_2(B) \, R_w(B) \, \Pi_B - \big( 1 \leftrightarrow 2 \big) \Bigg)
\end{align*}
Using cyclicity for the trace per unit volume and $[\Pi_B,S^z]=0$ (in view of~\eqref{sz is conserved}), one can then manipulate the above and get
\begin{equation} \label{derivative}
\begin{aligned}
& \frac{\di}{\di B} \, \IsDOS(B) \\
& \qquad = - \frac{1}{2\pi} \, \tau \left( \int_{\Gamma} \di w \, S^z \, \Pi_B \, R_w(B) \, P_1(B) \, R_w(B) \, P_2(B) \, R_w(B) \, \Pi_B - \big( 1 \leftrightarrow 2 \big) \right).
\end{aligned}
\end{equation}

An analogous expression is used in \cite{Moscolari} for the integrated density of states, obtained by replacing $S^z$ with $\Id$.
\end{remark}

\section{Spin Hall conductivity} \label{sec:SHC}

Next we analyze the response of a $2d$ electron gas, whose one-particle description is modeled by the Hamiltonian $H_B := H_{B,B}$ defined in \eqref{HamiltonianB_1B_2}, to an (alternating) electric current. We will be interested in particular in the response of a transverse spin current and in computing the associated spin conductivity, adapting the strategy of \cite{CorneanNenciuPedersen06}. As noted in the Introduction, when spin is not conserved the situation becomes more involved and this type of argument breaks down, thus requiring a new approach \cite{MarcelliPanatiTeufel}.

To compute the spin conductivity, we treat the electron gas as a grand canonical ensemble of non-interacting particles, initially (at time $t = - \infty$) in thermodynamic equilibrium. The gas is confined to a box $\Lambda_L:=\left(-L,L\right]^2$ of linear size $L \ge 1$, with Dirichlet boundary conditions; we will later be interested in the thermodynamic limit $L \to \infty$. The restriction of the Hamiltonian $H_B$ to the box with the specified boundary conditions is denoted as $H_B^{(L)}$. The asymptotic equilibrium state is modeled by the Fermi--Dirac distribution at temperature $T=\frac{1}{ \beta}$ and at chemical potential $\mu$ in the non-interacting spectral gap:
\begin{equation}
\label{FDistribution}
\rho(-\infty)=f\sub{FD}\big(H_B^{(L)}\big) , \quad
f\sub{FD}(x):= \frac{1}{\eu^{\beta(x-\mu)}+1} \, , \quad x \in \R\, ,\beta >0\, , \mu \in \R \, .
\end{equation}

The state then evolves adiabatically according to a time-dependent perturbed Hamiltonian
$$
H_{B}^{(L)}(t)=H_{B}^{(L)} + V(t) \, ,
$$
where the time-dependent electric potential is given by 
$$
V(\bx,t)= \left(\eu^{\iu \omega t}+\eu^{-\iu \bar{\omega}t}\right) E x_2 \, , \qquad t\leq 0 \, , \; \bx \in \Lambda_L \, , \; \Im\omega < 0\, , \; E \in \R .
$$
Notice that the role of the time-adiabatic parameter is played by the imaginary part of $\omega$.

The state of the system at $t \le 0$ is described by a density matrix $\rho(t)$ that satisfies the Liouville equation
$$
\iu \partial_t \rho(t) = \left[H_{B}^{(L)}(t),\rho(t)\right] \, , \quad \rho(-\infty)=f\sub{FD}\big(H_B^{(L)}\big) \, .
$$
Using the Dyson expansion in the interaction picture one can argue that the solution to the Liouville equation is given by 
$$
\rho(0)=f\sub{FD}\big(H_{B}^{(L)}\big) - \iu \int_{-\infty}^{0} \mathrm{d}s  \left[ V^{I}(s), f\sub{FD}\big(H_{B}^{(L)}\big) \right]  + \mathcal{O}(E^2) \, ,
$$
where $V^{I}(t):=\eu^{\iu t H_{B}^{(L)} } V(t) \eu^{-\iu t H_{B}^{(L)} }$. Assuming the validity of the linear response ansatz above, the spin current density that flows through the system at $t=0$ is given by
\begin{equation}
\label{CurrentDensity}
\begin{aligned}
\mathbf{j}^{z,(L)}&= \frac{1}{|\Lambda_L|} \Tr_{\mathcal{H}_L}\left( f\sub{FD}\big(H_{B}^{(L)}\big)  \mathbf{J}^z\sub{proper}\right) \\
&\phantom{=}- \frac{\iu}{|\Lambda_L|} \Tr_{\mathcal{H}_L} \left(  \int_{-\infty}^{0}  \mathrm{d}s  \left[ V^{I}(s), f\sub{FD}\big(H_{B}^{(L)}\big) \right]  \mathbf{J}^z\sub{proper}\right)  + \mathcal{O}(E^2)  \, .
\end{aligned}
\end{equation}
The conductivity tensor is therefore given by the first order coefficient of the spin current density expansion in $E$, namely
$$
{j}_\alpha^{z,(L)}=\left(\sigma_{\alpha,2}^{z,(L)}(\omega)+\sigma_{\alpha,2}^{z,(L)}(-\overline{\omega}) \right) E, \qquad \alpha \in  \left\{1,2\right\}.
$$
Since we are interested in the transverse conductivity, we consider only the first component of the spin current density, namely $j^{z,(L)}_1$. The equilibrium current, namely the first term in \eqref{CurrentDensity}, vanishes, as $f\sub{FD}\big(H_{B}^{(L)}\big)  J^z\sub{1,proper}= - \iu \big[H_{B}^{(L)}, f\sub{FD}\big(H_{B}^{(L)}\big) \, X_1\, S^z\big]$ and  the trace of a commutator is zero. Then, after integrating by parts and exploiting the cyclic properties of the trace, we get
$$
\begin{aligned}
&\sigma_{12}^{z,(L)}(B,T,\omega)=\frac{\iu}{\omega|\Lambda_L|}\Tr_{\mathcal{H}_L}\left(\left[X_2,f\sub{FD}\big(H_B^{(L)}\big)\right]\left[H_{B}^{(L)},X_1 S^z \right]\right) \\
&\phantom{=}-\frac{\iu}{\omega|\Lambda_L|} \Tr_{\mathcal{H}_L} \left( \int_{-\infty}^{0} \mathrm{d}s \,  \eu^{\iu s \left(H_{B}^{(L)}+\omega\right)} P_2(B) e^{-\iu s H_{B}^{(L)}} \left[    \left[H_{B}^{(L)},X_1S^z\right], f\sub{FD}\big(H_{B}^{(L)}\big) \right]   \right) \, .
\end{aligned} 
$$

Since $S^z$ is conserved, one can again use cyclicity of the trace to argue that the first term vanishes. In the second term, we can factor out the spin operator obtaining the so-called Kubo formula \cite{Kubo57}
$$
\begin{aligned}
&\sigma_{12}^{z,(L)}(B,T,\omega)\\ &=-\frac{1}{\omega|\Lambda_L|} \Tr_{\mathcal{H}_L} \left( S^z \int_{-\infty}^{0} \mathrm{d}s \, \eu^{\iu s \left(H_{B}^{(L)}+\omega\right)} P_2(B) e^{-\iu s H_{B}^{(L)}} \left[  P_1(B)  , f\sub{FD}\big(H_{B}^{(L)}\big) \right]   \right) \, .
\end{aligned}
$$

The previous formula can now be treated in the same way as in the charge current case, see \cite{CorneanNenciuPedersen06, CorneanNenciu09}. Due to the magnetic covariance of the operators in the trace, we first obtain that the spin conductivity admits a thermodynamic limit in the form of a trace per unit volume:
$$
\sigma_{12}^{z}(B,T,\omega):=\lim_{L \to \infty} \sigma_{12}^{z,(L)}(B,T,\omega).
$$
After that, we can perform the limit of zero temperature and zero frequency, where in particular the equilibrium state converges to the Fermi projection, obtaining a closed formula for the transverse spin conductivity which reads
\begin{equation}
\label{SpinConductivity}
\begin{aligned}
&\sigma_{12}^z(B):=\lim_{\omega \to 0}\lim_{T \to 0}\sigma_{12}^{z}(B,T,\omega) \\
&= - \frac{1}{4\pi} \tau \left( \int_{\Gamma} \di w \, S^z \, P_1(B) \, R_w(B) \, P_2(B) \, R^2_w(B) \right. -\left.S^z \, P_1(B) \, R^2_w(B) \, P_2(B) \, R_w(B) \vphantom{\int}\right).
\end{aligned}
\end{equation}

\vspace{-1em}

\section{Spin St\v{r}eda formula}

We are finally able to state our main result.

\begin{theorem}[Spin St\v{r}eda formula] \label{SSF}
For $\Pi_B := \Pi_{B,B}$
\begin{equation} \label{eqn:SSF}
\sigma_{12}^z(B) =  \,  \frac{\di}{\di B} \, \IsDOS(B) = \frac{1}{2\pi} \, \SCh(\Pi_B),
\end{equation}
where the {\it spin Chern marker} of $\Pi_B$ is defined as 
\begin{equation} \label{SCh}
\SCh(\Pi_B) := 2 \pi \iu \, \tau \big( S^z \, \Pi_B \, \big[ [X_1, \Pi_B], [X_2, \Pi_B] \big] \, \Pi_B \big). 
\end{equation}
\end{theorem}

\begin{proof}
Since $S^z$ commutes with $\mathbf{P}(B)$, the equality of the spin conductivity $\sigma_{12}^z(B)$ with the derivative of the integrated density of states is a simple consequence of the cyclicity for the trace per unit volume in the case of magnetic covariant operators with exponentially localized kernels, compare \eqref{4.21} and \eqref{SpinConductivity}.

Instead, the proof of the second equality claimed in Theorem \ref{SSF}, which relates the integrated spin density of states with the spin Chern marker, relies on the following general result from perturbation theory (see e.g.~\cite{Nenciu02}). To formulate its statement, recall that an operator $A$ is called (\emph{off}-)\emph{diagonal} with respect to a projection $\Pi = \Pi^2$ whenever 
\[ A = A\su{D} := \Pi \, A \, \Pi + \Pi^\perp \, A \, \Pi^\perp \quad \text{ (resp. } A = A\su{OD} := \Pi \, A \, \Pi^\perp + \Pi^\perp \, A \, \Pi \text{ )} \]
with $\Pi^\perp := \Id - \Pi$. It is clear that $A = A\su{D}$ if and only if $[A, \Pi]=0$; therefore, every operator of the form $[A,\Pi] = [A\su{OD},\Pi]$ is off-diagonal.

\begin{lemma} \label{lemma}
Let $A$ be an operator which is relatively bounded with respect to $H_B$. Define
 \[ \cI(A) := \frac{\iu}{2\pi} \oint_{\Gamma} \di w \, (H_B - w \Id)^{-1} \, A \, (H_B - w \Id)^{-1}, \]
 where $\Gamma$ is as in~\eqref{Riesz}. Then $C := \mathcal{I}(A)$ solves
 \begin{equation} \label{eqn:Liouvillian}
 \big[ H_B, C \big] = \big[ A, \Pi_B \big].
 \end{equation}
 Moreover, $\cI(A\su{OD}) = \cI(A)\su{OD}$ is the {\it unique} solution of the above equation among operators which are off-diagonal with respect to the orthogonal decomposition of $L^2(\R^2) \otimes \C^2$ induced by $\Pi_B$.
\end{lemma}
\begin{proof}
Observe first of all that $\cI(A)$ is a well-defined bounded operator, due to the relative boundedness of $A$ and to the boundedness of the resolvent. Consider now the identity of bounded operators \cite[Theorem 6.2.10]{ABG96} $R_w(B) \, [H_B, A] \, R_w(B) = [A,R_w(B)]$. 
Integrate both sides over $\Gamma$ and multiply by $\iu/2\pi$ to obtain
\[ \left[ H_B , \frac{\iu}{2\pi} \oint_{\Gamma} \di w \,  R_w(B) \, A \, R_w(B) \right] = \left[A,\frac{\iu}{2\pi} \oint_{\Gamma} \di w \, R_w(B)\right] = [A,\Pi_B] \]
where the last equality is the Riesz formula~\eqref{Riesz} for the projection $\Pi_B$. 

Since $[A,\Pi_B] = [A\su{OD},\Pi_B]$, $\cI(A\su{OD})$ solves \eqref{eqn:Liouvillian}, and moreover it is off-diagonal as the resolvent is a diagonal operator (it commutes with $H_B$ and therefore with its spectral projection $\Pi_B$), which implies that $\cI(A\su{OD}) = \cI(A)\su{OD}$. If now $C$ is another off-diagonal solution to \eqref{eqn:Liouvillian}, then $\cI(A\su{OD}) - C$ commutes with $H_B$ and is therefore both diagonal and off-diagonal. We conclude that $\cI(A\su{OD}) - C = 0$, which proves the uniqueness of $\cI(A\su{OD})$ among off-diagonal solutions to \eqref{eqn:Liouvillian}.
\end{proof}

Let us use the above Lemma to compute $\big[ [\Pi_{B}, A_1], [\Pi_{B}, A_2] \big]$ assuming that $A_1$ and $A_2$ are bounded off-diagonal operators with respect to $\Pi_{B}$. Dropping the subscript $B$, we have
\begin{align*}
& \big[ [\Pi, A_1], [\Pi, A_2] \big] = \left[ \left[ H, \frac{\iu}{2\pi} \oint_{\Gamma} \di w \, R_w \, A_1 \, R_w \right], \left[ H, \frac{\iu}{2\pi} \oint_{\Gamma'} \di w' \, R_{w'} \, A_2 \, R_{w'} \right] \right] \\
& = \left( \frac{\iu}{2\pi} \right)^2 \oint_{\Gamma} \di w \oint_{\Gamma'} \di w' \, \big\{ R_w \, [H,A_1] \, R_w \, R_{w'} \, [H,A_2] \, R_{w'} - R_{w'} \, [H,A_2] \, R_{w'} \, R_w \, [H,A_1] \, R_w  \big\}.
\end{align*}
We have chosen the contours $\Gamma, \Gamma'$ so that both surround the relevant spectral island of $H$ and $\Gamma'$ has a slightly smaller diameter than $\Gamma$. For future reference, let us denote by $D \subset \C$ the region surrounded by $\Gamma$, and by $D' \subset D$ the region surrounded by $\Gamma'$.

We further elaborate the above expression by using the resolvent identity $R_w \, R_{w'} = (w-w')^{-1} \, \big( R_w - R_{w'} \big)$. 
We conclude that $\big[ [\Pi, A_1], [\Pi, A_2] \big] = \mathrm{I} + \mathrm{II} + \mathrm{III} + \mathrm{IV}$ with
\begin{align*}
\mathrm{I} & := \left( \frac{\iu}{2\pi} \right)^2 \oint_{\Gamma} \di w \oint_{\Gamma'} \di w' \,  \frac{1}{w-w'} \,  R_w \, [H,A_1] \, R_w \, [H,A_2] \, R_{w'} , \\
\mathrm{II} & := -\left( \frac{\iu}{2\pi} \right)^2 \oint_{\Gamma} \di w \oint_{\Gamma'} \di w' \,  \frac{1}{w-w'} \,  R_w \, [H,A_1] \, R_{w'} \, [H,A_2] \, R_{w'} , \\
\mathrm{III} & := - \left( \frac{\iu}{2\pi} \right)^2 \oint_{\Gamma} \di w \oint_{\Gamma'} \di w' \, \frac{1}{w'-w} \, R_{w'} \, [H,A_2] \, R_{w'} \, [H,A_1] \, R_w , \\
\mathrm{IV} & := \left( \frac{\iu}{2\pi} \right)^2 \oint_{\Gamma} \di w \oint_{\Gamma'} \di w' \, \frac{1}{w'-w} \, R_{w'} \, [H,A_2] \, R_w \, [H,A_1] \, R_w .
\end{align*}

Let us show how to compute the term $\mathrm{I}$. We rewrite
\[ \mathrm{I} = \frac{\iu}{2\pi} \oint_{\Gamma} \di w \, R_w \, [H,A_1] \, R_w \, [H,A_2] \left( \frac{\iu}{2\pi} \oint_{\Gamma'} \di w' \,  \frac{R_{w'}}{w-w'} \right). \]
We multiply the integrand in parenthesis by $\Id = \Pi + \Pi^\perp$. Using that $R_w$ is a diagonal operator we can write
\[
\frac{\iu}{2\pi} \oint_{\Gamma'} \di w' \,  \frac{R_{w'}}{w-w'} = \frac{\iu}{2\pi} \oint_{\Gamma'} \di w' \, \frac{\Pi \, R_{w'} \, \Pi}{w-w'} + \frac{\iu}{2\pi} \oint_{\Gamma'} \di w' \, \frac{\Pi^\perp \, R_{w'} \, \Pi^\perp}{w-w'}.
\]

Let us look at the second term on the right-hand side of the above equality. Since $w \not\in D'$, the function $w' \mapsto \Pi^\perp \, R_{w'} \, \Pi^\perp /(w-w')$ is {\it holomorphic} in the region $D'$: indeed, the singularity of the resolvent at real energies $w'$ in the spectral island below the gap is removed by the presence of the projection $\Pi^\perp$ away from this spectral subspace. We conclude that its integral over $\Gamma'$ vanishes by Cauchy's theorem. As for the first integral, instead, we observe that the function $w' \mapsto \Pi \, R_{w'} \, \Pi$ is holomorphic in the {\it exterior} region $\C \setminus \overline{D'}$: this time the projection $\Pi$ removes the singularities of the resolvent above the spectral gap. Using Cauchy's integral formula for unbounded domains \cite[Problem 14.14]{Markushevich65}, we conclude that
\[ \frac{\iu}{2\pi} \oint_{\Gamma'} \di w' \, \frac{\Pi \, R_{w'} \, \Pi}{w-w'} = \frac{1}{2\pi\iu} \oint_{\Gamma'} \di w' \, \frac{\Pi \, R_{w'} \, \Pi}{w'-w} = \Pi \, R_\infty \, \Pi - \Pi \, R_w \, \Pi = - \Pi \, R_w \, \Pi, \]
where $\Pi \, R_\infty \, \Pi := \lim_{|z|\to\infty} \Pi \, R_z \, \Pi = 0$ since $\big\|\Pi \, R_w \, \Pi\big\| = 1/ \mathrm{dist}(z,\sigma(\Pi \, H \, \Pi))$. Finally
\[ \mathrm{I} = - \frac{\iu}{2\pi} \oint_{\Gamma} \di w \, R_w \, [H,A_1] \, R_w \, [H,A_2] \, R_w \, \Pi. \]

The other three terms, namely $\mathrm{II},\mathrm{III}$ and $\mathrm{IV}$, can be computed using the same strategy; moreover, we can freely change the integration contour from $\Gamma'$ to $\Gamma$ in the final expressions due to the analyticity of the resolvent operator in the spectral parameter. We conclude that $-\iu \, \big[ [\Pi, A_1], [\Pi, A_2] \big] = \big( \Pi \, T_{A_1,A_2} \, \Pi - \Pi^\perp \, T_{A_1,A_2} \, \Pi^\perp \big)$ where 
\begin{equation} \label{double commutator}
T_{A_1,A_2} :=\oint_{\Gamma} \di w \, \big\{ R_w \, \big( \iu [H,A_1] \big) \, R_w \, \big( \iu [H,A_2] \big) \, R_w - R_w \, \big( \iu [H,A_2] \big) \, R_w \, \big( \iu [H,A_1] \big) \, R_w \big\}
\end{equation} 
with $R_w = R_w(B)$. In particular $\Pi \, T_{A_1,A_2} \, \Pi = - \iu \, \Pi \, \big[ [\Pi, A_1], [\Pi, A_2] \big] \, \Pi$.

We apply now the above considerations to the operators  $A_j = [X_j, \Pi_{B}]$, $j \in \left\{1,2\right\}$, which are indeed bounded and off-diagonal with respect to $\Pi = \Pi_B$. Notice first of all that, by the Jacobi identity and the fact that $\dot{\mathbf{X}} = \mathbf{P}$, 
\[ \iu \big[ H_B, [X_j, \Pi_B] \big] = - \iu \big[ \Pi_B, [H_B, X_j] \big] = - \big[ \Pi_B, P_j(B) \big] \]
if all terms are sandwiched between $R_w(B)$. Plugging the above into \eqref{double commutator} and performing some simple algebra, we conclude that
\begin{align*}
\iu \, \Pi \, \big[ [X_1, \Pi] , [X_2, \Pi] \big] \, \Pi & = \Pi \, T_{[X_1, \Pi], \, [X_2,\Pi]} \, \Pi \\
& = - \frac{1}{2\pi} \oint_\Gamma \di w \, \Big\{ \Pi \, R_w \,  P_1 \, R_w \, P_2 \, R_w \, \Pi - \big( 1 \leftrightarrow 2 \big) \Big\} \\
& \quad + \frac{1}{2\pi} \oint_\Gamma \di w \, \Big\{ \Pi \, R_w \, \Pi \, P_1 \, \Pi \, R_w \, \Pi \, P_2 \, \Pi \, R_w \, \Pi - \big( 1 \leftrightarrow 2 \big) \Big\} .
\end{align*}
We now multiply the above operator identity by $S^z$ on the left and take the trace per unit volume. Observe that $S^z$ commutes with all the operators appearing on the third line of the above in view of~\eqref{sz is conserved}; the cyclicity property then implies that this line does not contribute to the trace per unit volume, being skew-symmetric in the exchange of the indices $(1 \leftrightarrow 2)$. Therefore by comparing the result equation to~\eqref{derivative} and~\eqref{SCh}, we see that this is exactly the claimed second equality in~\eqref{eqn:SSF}. This concludes the proof of Theorem~\ref{SSF}.
\end{proof}

\section{Quantization of the spin Hall conductivity}
\label{QuantSHall}

The spin St\v{r}eda formula~\eqref{eqn:SSF} can be used to deduce that the spin Hall conductivity $\sigma_{12}^z$ is {\it quantized} to a half-integer or integer in units of $ 1/2\pi$. Indeed, under the assumption~\eqref{sz is conserved} the Fermi projection splits as a direct sum of two projections, one corresponding to spin ``up'' and one corresponding to spin ``down'':
\[ \Pi_{B} = \Pi_{B}^{\up} \oplus \Pi_{B}^{\down}, \quad \Pi_{B}^{\up/\down} := P^{\up/\down} \, \Pi_{B} \, P^{\up/\down}, \quad P^{\up/\down} := \frac{1}{2} \Id \otimes ( \Id \pm \sigma^z ). \]
It is not difficult to see that then
\[ \SCh(\Pi_B) = \frac{1}{2} \Big( \Ch\big(\Pi_B^{\up}\big) - \Ch\big(\Pi_B^{\down}\big) \Big), \text{ where } \Ch(\Pi):= 2 \pi \iu \, \tau \big( \Pi \, \big[ [X_1, \Pi], [X_2, \Pi] \big] \, \Pi \big). \]
$\Ch(\Pi)$ is the {\it Chern marker} (sometimes called also {\it Chern character}) of the projection $\Pi$. It can be argued (see e.g.~\cite{CorneanMonacoMoscolari18}) that $\Ch(\Pi_B)$ is a {\it continuous} function of $B$ (as long as the spectral gap around the Fermi energy does not close), which attains {\it integer values}: therefore $\Ch\big(\Pi_B^{\up/\down}\big)$ must be both {\it constant integers} in $B$ if the spectral gap remains open, and we deduce that $\SCh(\Pi_B) \in \frac{1}{2} \Z$.

Notice that when $B_1=B_2=0$, then $H_{0,0}$ commutes with the fermionic time-reversal operator $\Theta := \iu (\Id \otimes \sigma^y) K$, where $K$ is the complex conjugation operator on $L^2(\R^2) \otimes \C^2$. This implies in particular that $\Ch\big(\Pi_0^{\down}\big) = \Ch\big(\Theta \, \Pi_0^{\up} \, \Theta^{-1}\big) = - \Ch\big(\Pi_0^{\up}\big)$, cf.~\cite{MonacoPanati15}. We therefore conclude that actually
\[ \SCh(\Pi_0) = \Ch\big(\Pi_0^{\up}\big) \in \Z. \]
The parity of this integer is the \emph{spin Chern number} proposed in~\cite{Prodan09} as a $\Z_2$-valued topological invariant to classify (disordered) time-reversal symmetric topological insulators. In the setting of discrete Bogoliubov--de Gennes Hamiltonians and with the assumption of conserved spin, the equality between the spin conductivity and the spin Chern number has been proved in \cite{DeNittisSchulzBaldes}. A study of the spin Chern number in relation to spin transport and its robustness to small perturbations that break spin conservation can be found in \cite{Prodan09, SchulzBaldes}, see also the discussion in \cite{MarcelliPanatiTeufel}.

\bigskip


\bigskip\bigskip

{\footnotesize
\begin{tabular}{rl}
 (D. Monaco) & \textsc{Dipartimento di Matematica, ``La Sapienza'' Universit\`{a} di Roma} \\
 &   Piazzale Aldo Moro 2, 00185 Roma, Italy \\
 &  \textsl{E-mail address}: \href{mailto:monaco@mat.uniroma1.it}{\texttt{monaco@mat.uniroma1.it}} \\
 \\
(M. Moscolari) &  \textsc{Department of Mathematical Sciences, Aalborg University} \\
&   Skjernvej 4A, 9220 Aalborg, Denmark \\
&  \textsl{E-mail address}:\href{mailto:massimomoscolari@math.aau.dk}{\texttt{massimomoscolari@math.aau.dk}}
\end{tabular}
}
\end{document}